\documentclass[journal,twoside]{IEEEtran}
\usepackage{cite}
\ifCLASSINFOpdf
\else
\fi
\usepackage[cmex10]{amsmath}
\usepackage{array}
\usepackage{url}
\usepackage{amsthm}
\usepackage{bbm}
\usepackage{amssymb}

\newtheorem{theorem}{Theorem}
\newtheorem{lemma}{Lemma}
\newtheorem{proposition}{Proposition}
\newtheorem{example}{Example}
\newtheorem{definition}[theorem]{Definition}
\newcommand{\1}{\mathbbm{1}}

\newcommand{\be}{\begin{equation}}
\newcommand{\bea}{\begin{eqnarray}}
\newcommand{\eea}{\end{eqnarray}}
\newcommand{\ee}{\end{equation}}

\def\id{{\rm id}}

\def\E{{\cal E}}
\renewcommand{\>}{\rangle}
\newcommand{\<}{\langle}

\def\C{\mathbb{C}}
\def\E{\mathcal{E}}
\def\1{\mathbbm{1}}
\def\MDD{{\cal M}_{D\times D}}
\def\N{\mathbb N}

\begin{document}
\title{A quantum version of Wielandt's inequality}

\author{Mikel~Sanz,
        David~P\'erez-Garc\'ia,
        Michael~M.~Wolf,
        and~Juan~I.~Cirac% <-this % stops a space
\thanks{Manuscript received September 30, 2009; revised .}% <-this % stops a space
\thanks{M.~Sanz and J.~I.~Cirac acknowledge the support of the QCCC Program of the
EliteNetzWerk Bayern and the DFG (project FOR635), D.~P\'erez-Garc\'ia acknowledges the support of the Spanish grants I-MATH, MTM2008-01366, S2009/ESP-1594 and the EU project QUEVADIS and M.~M.~Wolf acknowledges support by QUANTOP, the Danish Natural Science Research Council (FNU) and the EU projects QUEVADIS and COQUIT.}%
\thanks{M.~Sanz and J.~I.~Cirac are with the Max-Planck-Institut f\"ur Quantenoptik, Hans-Kopfermann-Str. 1, 85748 Garching, Germany.}% <-this % stops a space
\thanks{D.~P\'erez-Garc\'ia is with the Dpto. An\'alisis Matem\'atico and IMI, Universidad Complutense de Madrid, 28040 Madrid, Spain.}% <-this % stops a space
\thanks{M.~M.~Wolf is with the Niels Bohr Institute, Blegdamsvej 17, 2100 Copenhagen, Denmark.}}

\maketitle

\begin{abstract}
%\boldmath
In this paper, Wielandt's inequality for classical channels is extended to quantum channels.
That is, an upper bound to the number of times a channel must
be applied, so that it maps any density operator to one with full
rank, is found. Using this bound, dichotomy theorems for the
zero--error capacity of quantum channels and for the Matrix Product
State (MPS) dimension of ground states of frustration-free
Hamiltonians are derived. The obtained inequalities also imply new bounds on the
required interaction-range of Hamiltonians with unique MPS ground
state.
\end{abstract}


% \fi
\IEEEpeerreviewmaketitle

\section{Introduction}
\IEEEPARstart{C}{onsider} a classical memoryless channel acting in discrete time on
an alphabet of size $D$. Such a channel is described by a stochastic
matrix $A\in \MDD$ which is called \emph{primitive} \cite{Horn} if
there is an $n\in \N$ such that $(A^n)_{i,j} > 0$ for all $i,j$. The minimum $n$
for which this occurs, $p(A)$, is called the (classical) \emph{index
of primitivity} of $A$ (or the \emph{exponent} of $A$). This ensures
that after applying the channel $p(A)$ times to any probability
distribution, there will be a non--zero probability for any possible
event. \emph{Wielandt's inequality} \cite{Wielandt} states that, for
every primitive matrix,
\begin{equation*}
p(A)\le D^2 - 2D + 2
\end{equation*}
and this is the optimal bound which is independent of the matrix
elements. Wielandt's
inequality has a wide range of applications in different fields,
ranging from Markov chains \cite{Seneta} to graph theory and number
theory \cite{Frobenius}, and numerical analysis \cite{Varga}.

In this work, we derive a quantum analogue of Wielandt's inequality.
That is, we consider \emph{quantum channels}, i.e., trace preserving
completely positive linear maps (TPCPM) and define a property
analogous to primitivity: the existence of an $n\in\mathbb{N}$ such
that after the $n$-fold application of the quantum channel every
positive semidefinite operator is mapped onto a  positive definite
operator. The smallest such $n$ then defines a quantum index of
primitivity, $q$. We begin by relating primitivity to some other
properties of quantum channels, such as the existence of a unique
full-rank fixed point or the fact that the Kraus operators
corresponding to some number of applications of the channel span the
full space of matrices. This will allow us to derive a quantum
Wielandt's inequality for primitive quantum channels,
\begin{equation*} q\le
(D^2 - d + 1)D^2
\end{equation*}
where $D$ is the dimension of the Hilbert space, and $d$ the
number of linearly independent Kraus operators. We will see that,
under certain generic conditions on the Kraus operators, better
inequalities can be derived. Finally, we apply the new inequalities
to three problems related to channel capacities and to quantum spin
chains:  we derive a dichotomy theorem for the zero--error capacity
of quantum channels and prove a conjecture for ground states of
frustration-free spin Hamiltonians. Moreover, we show that our
result also has new implications concerning the interaction--range
of Hamiltonians with MPS as unique ground states \cite{PVWC07}.

% SECTION
% SECTION

\section{Basic notions}
Let us start by fixing the notation and
introducing the basic notions. We will consider quantum channels,
i.e., TPCPMs,  ${\cal E}:\MDD \to \MDD$, where $\MDD$ is the space
of complex $D \times D$ matrices. Let us denote by ${\cal E}_A$ the
quantum channel with Kraus operators $\{A_k\in\MDD\}_{k=1}^d$, i.e.,
\begin{equation}
 {\cal E}_A(X)=\sum_{k=1}^d A_k X A_k^\dagger.
\end{equation}
We define $S_n(A)\subseteq \MDD$ as the linear space spanned by all
possible products of exactly $n$ Kraus operators, $A_{k_1} A_{k_2}
\ldots A_{k_n}$, and by $A^{(n)}_k$ the elements of $S_n(A)$. There
is a one-to-one correspondence between a quantum channel $\E$ and
its Choi matrix $\omega(\E):=(\id\otimes\E)(\Omega)$ where
$\Omega=\sum_{i,j=1}^D|ii\>\<jj|$. It is readily verified that ${\rm
rank}\big[\omega({\E_A}^n)\big]=\dim\big[S_n(A)\big]$. We further
define $H_n(A,\varphi):= S_n(A)|\varphi\rangle \subseteq \C^D$ as
the space spanned by all vectors $A_{k_1} A_{k_2} \ldots
A_{k_n}|\varphi\rangle$, where $|\varphi\rangle\in\C^D$. That is,
${\rm
rank}\big[{\E_A}^n(|\varphi\>\<\varphi|)\big]=\dim\big[H_n(A,\varphi)\big]$.

We introduce now three properties which will later turn out to be
equivalent:
\begin{description}
  \item[(a)] A quantum channel  ${\cal E}_A$ is called {\em primitive} if
there exists some $n\in \N$ such that for all
$|\varphi\rangle\in \C^D$, $H_n(A,\varphi)=\C^D$. In other
words, if for every input density operator $\rho$ the output
${\E_A}^n(\rho)$ obtained after $n$ applications of the channel
has full-rank. We will denote by $q(\E_A)$ the minimum $n$ for
which that condition is fulfilled. Note that if ${\cal E}_A$ is
primitive, then for every $m\in\mathbb{N}$, ${\cal E}_A^m$ is
primitive, too, and  we have $H_n(A,\varphi)=\C^D$ for all $n\ge
q(\E_A)$.
  \item[(b)] A quantum channel ${\cal E}_A$ is said to have \emph{eventually
full Kraus rank} if there exists some $n\in \N$ such that
$S_n(A)=\MDD$, i.e., if ${\rm
rank}\big[\omega({\E_A}^n)\big]=D^2$. We denote by $i(A)$ the
minimum $n$ for which that condition is satisfied. Obviously, if
${\cal E}_A$ fulfills this property, then $S_n(A)=\MDD$ for all
$n\ge i(A)$.
  \item[(c)] We say that a quantum channel ${\cal E}_A$ is {\em strongly
irreducible}\footnote{The notion of 'irreducibility', used for
instance in \cite{gen-fro}, differs from our definition of
'strong irreducibility' by allowing for other eigenvalues of
magnitude one. In fact, $\E$ is strongly irreducible iff $\E^n$
is reducible for all $n\in\mathbb{N}$. This property is known as
\emph{injectivity} in the context of Matrix Product States.} if
the following two conditions are fulfilled: (i) ${\cal E}_A$ has
a unique eigenvalue, $\lambda$, with $|\lambda|=1$; (ii) the
corresponding eigenvector, $\rho$, is a  positive
definite operator ($\rho > 0$). This implies the convergence
 \begin{equation}
 \label{Einfty} \lim_{n\to\infty} {\cal E}_A^n = {\cal
 E}_A^\infty,
 \end{equation}
where ${\cal E}_A^\infty(X) := \rho \, {\rm tr}(X)$. Note that,
for instance, the generalized Frobenius theorem proved in
\cite[Theorem 2.5]{gen-fro} ensures that a TPCPM always has an
eigenvalue $\lambda=1$ with eigenvector $\rho\ge 0$. In this
case it was already known \cite[Lemma 5.2]{FNW92} that there
exits an upper bound for $i(A)$ related with the second
eigenvalue $\lambda_2$ of $\E_A$, which is essentially
$i(A)\lesssim O(\exp \frac{1}{\lambda_2})$.
\end{description}

Our first simple observation is that (b) implies (a), or stated
quantitatively:
\begin{proposition}\label{prop:iq}
For every quantum channel $\E_A$ we have that $q(\E_A)\leq i(A)$.
\end{proposition}
\begin{proof}
Take any $n\geq i(A)$. Then, by the definition of $i(A)$, the Choi
matrix $\omega(\E_A^n)$ has full-rank, so that
\begin{equation*}
\E_A^n(|\varphi\>\<\varphi|)=\;(\1\otimes\<\bar\varphi|)\;
\omega(\E_A^n)\;(\1\otimes|\bar\varphi|)
\end{equation*}
also has full-rank.
\end{proof}

Before continuing the analysis of the relationships among the three
properties above in the quantum context, let us connect them to the
classical notion of primitivity.  Given a stochastic matrix $A =
(a_{ij})$, let us consider the map ${\cal E}_A$ defined by the Kraus
operators $A_{i,j}=\sqrt{a_{i,j}} |i\rangle\langle j|$. ${\cal E}_A$
has the property that for an operator $\rho$ with entries
$\rho_{i,j}=\delta_{i,j} p_i\ge 0$, $\rho':={\cal E}_A (\rho)$ is
diagonal with $\rho'_{i,j}=\delta_{i,j} p'_i$, with $p'=Ap$. Thus,
${\cal E}_A$ implements the stochastic map $A$, i.e., the quantum
channel reduces to the classical channel when applied to diagonal
density operators. Note that $d$ is the number of positive entries
of the stochastic matrix in the classical case, so the general quantum bound
applied to a classical channel is always worse than the classical
bound.

Let us consider $A$ primitive and denote by $p(A)$ its classical
index of primitivity. Then, we have:

\begin{proposition}
Let us consider a primitive stochastic map $A$ and the corresponding
TPCPM ${\cal E}_A$. Then, $\E_A$ is also primitive and the equality
$q(A)=p(A)=i(A)$ holds.
\end{proposition}

\begin{proof}
It is clear that $p(A)\le q(\E_A)$ and we proved in Prop.
\ref{prop:iq} that $q(\E_A)\le i(A)$. In order to show that $i(A)\le
p(A)$, we define $\tilde A_{i,j}=\sqrt{a_{i,j}} A_{i,j}$, $n =
p(a)-1$, and take
 \begin{equation}
 \sum_{k_1,\ldots,k_n=1}^D \tilde A_{i,k_1} \tilde A_{k_1,k_2}\ldots
 \tilde A_{k_n,j}= (A^{p(a)})_{i,j} |i\rangle\langle j|\ne 0.
 \end{equation}
Thus, $|i\rangle\langle j|\in S_{p(A)}(A)$ for all $i,j$.
\end{proof}

We note that $q(\E_A)$ is different from $i(A)$ in the general case.
To see that, let us consider an example with $d=3$, $D=2$, and take
as Kraus operators $\sigma_\alpha/\sqrt{3}$, where $\alpha=x,y,z$
labels the three Pauli matrices. Here $q(\E_A) = 1 < i(A) = 2$.

However, the following proposition shows that $i(A)$ is finite
whenever $q(\E_A)$ is. In fact, all three definitions above are
equivalent:
\begin{proposition}\label{prop:equiv}
Given a quantum channel ${\cal E}_A$, the following statements are
equivalent: (a) ${\cal E}_A$ is primitive; (b) ${\cal E}_A$ has
eventually full Kraus rank; (c) ${\cal E}_A$ is strongly
irreducible.
\end{proposition}

\begin{proof}
We denote by $\rho\ge 0$ an eigenoperator of ${\cal E}_A$
corresponding to the eigenvalue $\lambda=1$.
\begin{trivlist}
\setlength\leftmargin {0em}
  \item[(b) $\Rightarrow$ (a)] \mbox{}\par
This implication is given by Prop.\ref{prop:iq}.
  \item[(a) $\Rightarrow$ (c)] \mbox{}\par
We prove it by contradiction. Let us assume that ${\cal E}_A$ is not
strongly irreducible. Then, we must have at least one of the
following cases: (i) $\rho$ is not full-rank; (ii) there exists
another eigenoperator, $\rho'$, corresponding to $\lambda=1$; (iii)
there is another eigenvalue, $\lambda'$, with $|\lambda'|=1$. Since
for all $n\in N$, ${\cal E}_A^n(\rho)=\rho$, (i) automatically
implies that ${\cal E}_A$ is not primitive. Furthermore, if we have
(ii),  choosing $\epsilon=1/\max[{\rm spec}(\rho^{-1/2}{\rho'}\rho^{-1/2})]$ we have that $\tilde
\rho=\rho-\epsilon \rho'\ge 0$ is not full-rank and thus we are back
in (i). Moreover, it is proven in the demonstration of
\cite[Proposition 3.3]{FNW92} that, if (i) and (ii) do not hold, the
other possible eigenvalues of modulus 1 are the $p$-th roots of
unity for some finite $p\in \N$. Therefore, we have (ii) for ${\cal
E}_A^p$, and thus ${\cal E}_A^p$ cannot be primitive.
  \item[(c) $\Rightarrow$ (b)] \mbox{}\par
This implication can be deduced from \cite[Lemma 5.2]{FNW92}, but we
include here a proof for completeness. We prove it by contradiction.
Let us assume that ${\cal E}_A$ is (j) strongly irreducible, but
(jj) does never get full Kraus rank. If we have (j) then $\rho$ is
full-rank and Eq. (\ref{Einfty}) is fulfilled. Because of (jj), for
all $n\in \N$ and $A^{(n)}_k\in S_n(A)$, there exists some $B_n\ne
0$ such that ${\rm tr}(A^{(n)}_k B_n)=0$. Thus,
 $$
 \left|{\rm tr}(\rho B_n^\dagger B_n)\right|=
 \left|\sum_{k_1,\ldots, k_n} |{\rm tr}(A_{k_1}\cdots A_{k_n} B_n)|^2 - {\rm tr}(\rho B_n^\dagger B_n)\right|$$
 $$=
 \left|{\rm tr}\left[\Omega({\cal E}_A^n\otimes \id)
 (\tilde{B}_n\Omega \tilde{B}_n^\dagger)\right]-{\rm tr}\left[\Omega({\cal E}_A^\infty\otimes \id)
 (\tilde{B}_n\Omega \tilde{B}_n^\dagger)\right]\right|$$
 $$\le c_n\|\Omega\|_\infty {\rm tr}(\tilde{B}_n\Omega \tilde{B}_n^\dagger)=
 D c_n {\rm tr}(B_n^\dagger B_n)
 $$
where $\tilde{B}_n=B_n\otimes\1$ and $\lim_nc_n=0$. If $\rho$ was
full-rank, then for all $X\ge 0$ one would have
$${\rm tr}(\rho X)\ge\frac{1}{\|\rho^{-1}\|_\infty} {\rm tr}(X)$$ and we obtain a contradiction.
\end{trivlist}
\end{proof}

As a consequence of Prop. \ref{prop:equiv}, we obtain that
primitivity of a quantum channel can be decided by observing its
spectral properties. In fact, this  is the precise quantum analogue
of the classical result that a stochastic matrix is primitive iff it
has a unique eigenvalue of maximum modulus and a  positive definite
fixed point (cf. \cite{Horn}).

% SECTION
% SECTION

\section{Quantum Wielandt's inequalities}

In order to reach a quantum version of Wielandt's inequality, i.e.
bounds for $q(\E_A)$ and $i(A)$, we require some preliminary lemmas:

\begin{lemma}
\label{lemma1} Let ${\cal E}_A$ be a primitive quantum channel on
$\MDD$ with $d$ Kraus operators. Then, there is a $A^{(n)}\in
S_n(A)$ with $n\le D^2-d+1$ such that ${\rm tr}(A^{(n)})\ne 0$.
\end{lemma}

\begin{proof}
Let us denote by $T_n(A)$ the span of all $S_m(A)$ with $m\le n$. We
just have to show that: (*) for any $n\in \N$, if ${\rm
dim}[T_n(A)]<D^2$, then ${\rm dim}[T_{n+1}(A)]>{\rm dim}[T_{n}(A)]$.
Since ${\rm dim}[T_1(A)]=d$, by iteration we obtain that $T_{D^2 - d
+ 1} = \MDD$. This implies that a linear combination of the elements
of $S_n(A)$ with various $n\le D^2-d+1$ must be equal to the
identity, and thus at least one of the elements must have non--zero
trace. To prove (*) we note that, by definition, $T_{n}(A)\subseteq
T_{n+1}(A)$. If they would be equal, then $T_{m}(A)= T_{n}(A)$ for
all $m>n$. Thus, ${\rm dim}[T_n(A)]=D^2$ since otherwise the map
${\cal E}_A$ would not be primitive.
\end{proof}

\begin{lemma}\label{lemma2}
Let ${\E}_A$ be primitive such that $A_1|\varphi\rangle =
\mu|\varphi\rangle$ with $\mu\ne 0$. Then: (a)
$H_{D-1}(A,\varphi)=\C^D$. (b) If $A_1$ is not invertible, then for
all $|\psi\rangle\in \C^D$, $|\varphi\rangle\langle\psi|\in
S_{D^2-D+1}(A)$;
\end{lemma}

\begin{proof}
(a) We define  $K_n(A,\varphi)$ as the span of all $H_m(A,\varphi)$
with $m\le n$ together with $|\varphi\rangle$. If ${\rm
dim}[K_n(A,\varphi)]<D$, then ${\rm dim}[K_{n+1}(A,\varphi)]>{\rm
dim}[K_{n}(A,\varphi)]$, since otherwise the map would not be
primitive. Thus, $K_{D-1}(A,\varphi)=\C^D$. That is, for all
$|\phi\> \in \C^D$, there exist matrices $A^{(n)}\in S_{k_n}(A)$,
$k_n\le D-1$ such that (with $A^{(0)}\propto\1$)
 \begin{equation}
 |\phi\rangle = \sum_{n=0}^{D-1} A^{(n)} |\varphi\rangle=
 \sum_{n=0}^{D-1} A^{(n)} \frac{A_1^{D-k_n}}{\mu^{D-k_n}}|\varphi\rangle,
 \end{equation}
and thus, $|\phi\rangle\in H_{D-1}(A,\varphi)$. (b) We write $A_1$ in
the Jordan standard form and divide it into two blocks. The first
one, of size $\tilde D\times \tilde D$, consists of all Jordan
blocks corresponding to non--zero eigenvalues, whereas the second
one contains all those corresponding to zero eigenvalues. We denote
by $P$ the projector onto the subspace where the first block is
supported and by $r\le
D-\tilde D$ the size of the largest Jordan block corresponding to a
zero eigenvalue. We have
 \begin{equation}
 \label{propA1}
 A_1 P = P A_1, \quad A_1^r =A_1^r P.
 \end{equation}
We define $R_n(A)=PS_n(A)$ and show that $R_{D\tilde D}(A)={\cal
M}_{\tilde D\times D}$. For all $n\in \mathbb{N}$, ${\rm
dim}[R_{n+1}(A)]\ge{\rm dim}[R_{n}(A)]$. The reason is that for any
linearly independent set of matrices $A^{(n)}_k\in R_n(A)$,
$A_1A^{(n)}_k\in R_{n+1}(A)$ are also linearly independent, given
that $A_1$ is invertible on its range. By following the reasoning of
\cite[Appendix A]{PVWC07} we get that, if ${\rm
dim}[R_{n+1}(A)]={\rm dim}[R_{n}(A)]=:D'$, then ${\rm
dim}[R_m(A)]=D'$ for all $m>n$, which is incompatible with ${\cal
E}_A$ being primitive unless $D'=\tilde D D$. Thus, for all
$|\psi\rangle\in \C^D$, there exists $A\in S_{\tilde D D}$ with
$|\varphi\rangle\langle \psi|= PA= A_1^r P A/\mu^r = A_1^r A/\mu^r
=A'\in S_{\tilde D D+r}$. By using that $\tilde{D}\leq D-r$ and that
$r\geq 1$ (since $A_1$ is supposed to be not invertible) we get
$\tilde D D+r\le D^2-D+1$, which concludes the proof.
\end{proof}

We have now the necessary tools to prove our main result.

\begin{theorem}\label{thm:mainthm}
Let ${\cal E}_A$ be a primitive quantum channel on $\MDD$ with $d$
Kraus operators. Then $q(\E_A)\leq i(A)$ and
\begin{enumerate}
  \item in general $i(A)\leq (D^2-d+1)D^2$,
  \item if the span of Kraus operators $S_1(A)$ contains an
  invertible element, then $i(A)\leq D^2-d+1$,
  \item if the span of Kraus operators $S_1(A)$ contains a
  non-invertible element with at least one non-zero eigenvalue, then $i(A)\leq D^2$.
\end{enumerate}
\end{theorem}

\begin{proof}
The inequality $q(\E_A)\leq i(A)$ is shown in Prop. \ref{prop:iq}.

$2.$ If there is an invertible element, then it follows from
\cite[Appendix A, Proposition 2]{PVWC07}that $\dim S_{n+1}(A)>\dim
S_n(A)$ until the full matrix space $\MDD$ is spanned and thus
$i(A)\leq D^2-d+1$.

$1.$ Let us denote by $\{A^{(n)}_k\}$ the Kraus operators
corresponding to ${\cal E}_A^{n}$. According to Lemma \ref{lemma1},
one of them, say $A^{(n)}_1$, has non--zero trace and therefore
there exists $|\varphi\>$ such that
$A^{(n)}_1|\varphi\>=\mu|\varphi\>$ with $\mu\not = 0$. If
$A^{(n)}_1$ is invertible, then $1.$ is implied by $2.$, so we can
assume that it is not invertible. According to Lemma
\ref{lemma2}.(b), for all $|\psi\rangle,|\chi\rangle\in \C^D$ we
have $|\varphi\rangle\langle \psi|\in S_{D^2-D+1}(B)$; and
according to Lemma \ref{lemma2}.(a) $|\chi\rangle\langle \psi|\in
S_{D^2}(B)= S_{n D^2}(A)$. This implies that
$S_{n D^2}(A)=M_{D\times D}$ and hence the general bound $1.$
follows.

The argument which proves $3.$ is completely analogous. The main
difference is that, in order to guarantee the existence of a Kraus
operator with non-zero eigenvalue, we have to apply Lemma
\ref{lemma1} for the general case $1.$ and to take the $n$'th power
of the quantum channel for some $n \le D^2-d+1$.
\end{proof}

We do not know whether, or in which cases, our bounds are sharp. A
simple lower bound to $i(A)$ comes from the examples showing that
the classical Wielandt's inequality is sharp. In these cases
$q(\E_A)=i(A)=D^2-2D+2$. A lower bound that goes beyond this value is given by the next example.

\begin{example}
Let us consider the completely positive map described by the following Kraus operators $A_i \in
\mathcal{M}_D$: $A_0 = \sum_{i=0}^{D-1} |i+1\>\<i|$ and $A_1 =
|1\>\<D-1|$, with $|D\> = |0\>$. In this case $i(A) = D^2 - D$ which is larger than the bound appearing in Wielandt's classical inequality whenever
$D>2$.
\end{example}
\begin{proof}
Consider the case $D>2$ as $D=2$ is readily verified by inspection. Then $A_1^2=0$ and $A_1A_0^kA_1=A_1\delta_{k,D-2}$. Therefore \be S_N(A)={\rm span}\{A_0^N,A_0^kA_1A_0^l\}, \ee
where $k,l=0,\ldots,D-1$ fulfill the additional constraint that \be\label{eq:const} k+l+1+n(D-1)+m D=N \ee for some $n,m\in\mathbb{N}_0$. The additional constraint comes from the fact that $A_1$ can stem from $A_1A_0^{D-2}A_1$ or $A_1A_0^{D-2}A_1A_0^{D-2}A_1$ etc. which is a monomial of degree $1+n(D-1)$. The fact that $A_0^{mD}=\1$ is taken care of by the additional factor $mD$.  Now assume that $N=D(D-1)-1$. Let us upper bound the number of linearly independent operators in $S_N(A)$. Clearly, for every chosen $n$ and $k$, we get that  $l$ and $m$ are fixed by the additional constraint. For $n=D-1$, the range of $k$ is by Eq.(\ref{eq:const}) restricted to $k=0,\ldots, D-3$. So in total we have at most $(D-2)+(D-1)D+1=D^2-1$ independent elements which  cannot span the entire matrix algebra. Thus $i(A)\geq D^2-D$ (if the map is primitive). That this bound is sharp, and the map actually primitive, is seen by noting that for $N=D^2-D$ the constraint in Eq.(\ref{eq:const}) allows us to choose $k$ and $l$ freely by adjusting $n$ and $m$. Then, however, $A_0^kA_1A_0^l$ runs through all matrix units which span the entire matrix algebra.
\end{proof}
We also note that for small dimension $D=2,3$ there is always an element in $S_1(A)$ which has a non-zero eigenvalue. In other words, in these cases the first bound in Thm.\ref{thm:mainthm} never applies without one of the other bounds. The fact that $S_1(A)$ has this property for $D=2,3$ stems from the classification of nilpotent subspaces \cite{nilpotents}: assume that $S_1(A)$ would be a nilpotent subspace within the space of $D\times D$ matrices. Then for $D=2$ its dimension would have to be one, so it could not arise from the Kraus operators of a quantum channel. Similarly, for $D=3$ there are (up to similarity transformations) two types of nilpotent subspaces \cite{nilpotents} with $d>1$: one of dimension $d=3$, the space of upper-triangular matrices, whose structure does not allow the trace-preserving property, and one of dimension $d=2$ which only leads to quantum channels having a (in modulus) degenerate largest eigenvalue. Hence, if $S_1(A)$ is generated by the Kraus operators of a primitive quantum channel, then it cannot be nilpotent if $D=2,3$.

In the following we will show some applications of the derived
bounds.

% SECTION
% SECTION

\section{Zero-error capacity}

The zero error capacity $C_0$ of a noisy channel was defined by
Shannon in \cite{Shannon} as follows: There exists a sequence of
codes of increasing block length such that the
rate of transmission approaches $C_0$ and the probability of error
after decoding  {\it is} zero (instead of {\it
approaches zero} as in the definition of the usual capacity).
Furthermore, this is not true for any value higher than $C_0$. This
concept becomes important in situations where no error can be
tolerated or when a fixed finite number of uses of the channel is
available and it constitutes a central topic in information theory
\cite{zero-review}. The definition can be translated
straightforwardly to the case of quantum channels \cite{Medeiros},
where a number of interesting results appear: the computation of
this is QMA-hard \cite{Shor} and it can be superactivated
\cite{Cubitt} (see also \cite{Duan,Duan1}).

We will show here a dichotomy behavior for the power of a quantum
channel with full-rank fixed point (e.g., a unital quantum channel)
as a consequence of our quantum Wielandt's inequality. If we think
of the power $\E^n$ as a channel describing the input-output
relation after $n$ units in time/space, then the subsequent result
shows that there is a critical time/length $n=q(\E_A)$ such that a
successful transmission through $\E^n$ implies the possibility of a
successful transmission to arbitrary $m\geq n$. By the quantum Wielandt's inequality, this critical value can be taken $(D^2-d+1)D^2$ and is therefore {\it universal}. It depends only on $D$ and not on the channel itself.

\begin{theorem}\label{thm:zero}
If $\E$ is a quantum channel with a full-rank fixed point, and we
call $C_0(\E)$ the 0-error-classical capacity of $\E$. Then,
either\footnote{In fact, one can consider here even the one-shot zero-error
capacity, that is, the one obtained with a single use of the channel.} $C_0(\E^n)\ge 1$ for all $n$ or $C_0(\E^{q(\E)})=0$.
\end{theorem}
\begin{proof}
We split the problem into two cases:
\begin{description}
  \item[Case 1: ] \mbox{}\par
Let us assume that the channel has two (or more) different fixed
points. By following \cite{Lindblad}, the set of fixed points of
a quantum channel which has a full-rank fixed point is of the
form $V\left(\oplus_i \rho_{i}\otimes{M_{m_i}} \right)V^\dagger$
where $V$ is some unitary, the $\rho_i$'s are density matrices
with full-rank, and ${M_{m_i}}$ is a full matrix algebra of
dimension $m_i$. Consequently, if the direct sum is non-trivial,
we can encode a classical bit in the corresponding projectors.
If the direct sum is trivial, then the space of matrices is
non-trivial, i.e., there is a $m_i\geq 2$, and we can encode one
qubit in it. In either case $C_0(\E^n)\ge 1$ independent of $n$.

Similar statements hold if the channel has only one fixed point
(which is by assumption full-rank) but another eigenvalue $\mu$
of magnitude one: since $\mu$ is a root of unity, i.e., there is
an integer $p\leq D^2$ with $\mu^p=1$, we have that $\E^p$ has
several fixed points. So again we can safely encode a bit and
$C_0(\E^n)\ge 1$ independent of $n$.
  \item[Case 2: ] \mbox{}\par
If the channel has just one fixed point and no other eigenvalue
of magnitude one, then it is primitive by Prop.\ref{prop:equiv}.
So $\E^{q(\E)}$ has the property that all output states are
full-rank. This implies \cite{Cubitt} that $C_0(\E^{n})=0$ for
all $n\geq q(\E)$.
\end{description}
\end{proof}

% SECTION
% SECTION

\section{Frustration-free Hamiltonians and Matrix Product
States}

Matrix Product States have proven to be a useful family of quantum
states for explaining the low energy sectors of locally interacting
one-dimensional systems. They constitute a suitable variational
ansatz for instance to compute ground state energies to high
accuracy \cite{Frank-advances} which can be explained by the fact
that MPS approximate ground states of local 1D Hamiltonians well
\cite{Hastings}. Similarly they are used to understand effects on 1D
quantum systems on analytic grounds, such as string orders
\cite{PWSVC08}, symmetries \cite{SWPC09}, renormalization flows
\cite{renormalization} or sequential interactions \cite{sequential,sequential2}.

Associated to each translational invariant MPS of the form
\begin{equation}\label{eq:TI-MPS}
|\psi_A\>=\sum_{i_1,\ldots, i_N}{\rm tr}(A_{i_1}\cdots A_{i_N})|i_1\cdots i_N\>
\end{equation}
there is a {\it parent Hamiltonian} $H_A$ which is frustration-free
and has $|\psi_A\>$ as ground state. Let us start by defining the
concept of frustration-free Hamiltonian. Consider a local
translational invariant Hamiltonian in a spin chain $H=\sum_i
\tau^i(h)$ where $h$ denotes the local interaction term and $\tau$
the translation operator. Then,
\begin{definition}
The Hamiltonian is called \emph{frustration-free} if its ground
state $|\psi_0\>$ minimizes the energy {\it locally}, that is, if
\begin{equation}\label{ffre}
\min_{|\psi\>}\<\psi|\1\otimes h|\psi\>=\<\psi_0|\1\otimes h|\psi_0\>.
 \end{equation}
 \end{definition}
We assume w.l.o.g. that (\ref{ffre}) is equal to $0$.  Such
Hamiltonians include classical Hamiltonians, where the terms
commute, as well as all {\it parent Hamiltonians} appearing in the
Matrix Product State (MPS) theory \cite{FNW92,PVWC07,SCP10}. A
remarkable example is the AKLT Hamiltonian \cite{AKLT88}.

The corresponding local interaction term $h$ above is constructed as
the projector onto the orthogonal complement of the image of
\begin{equation}
\label{eq:gamma}X\in M_{D\times D}\mapsto
\sum_{i_1,\ldots, i_L}{\rm tr}(XA_{i_1}\cdots A_{i_L})|i_1\cdots
i_L\>,
\end{equation}
for some sufficiently large \emph{interaction range} $L$. Note that
the map in Eq.(\ref{eq:gamma}) is \emph{injective} for sufficiently
large $L$ iff the map $\E_A(X):=\sum_iA_iXA_i^\dagger$ is
primitive\footnote{$\E_A$ may be assumed to be trace-preserving
without loss of generality \cite{PVWC07}.} and that injectivity holds for all
$L\geq i(A)$. The following theorem which was proven in
\cite{FNW92,PVWC07} provides another application for the quantum
Wielandt's inequalities for $i(A)$:
\begin{theorem}
If the interaction range $L$ of the parent Hamiltonian $H_A$
satisfies $L>i(A)$, then the MPS $|\psi_A\>$ is the unique ground of
$H_A$, and $H_A$ has a spectral gap above the ground state energy.
\end{theorem}
Hence, the quantum Wielandt's inequality provides a bound for the
interaction length required to get a {\it good} parent Hamiltonian
for a MPS. Indeed, the existence of such inequality was already
conjectured in the context of MPS \cite[Conjectures 1 and 2]{PVWC07}
and some results obtained so far about MPS do directly depend on the
validity of that conjecture. In particular, a dichotomy result for
ground states of frustration-free Hamiltonians, sketched in
\cite{PVWC07} and for which we give a complete proof below, and the
characterization of the existence of global symmetries in arbitrary
MPS given in \cite{SWPC09}.

One might conjecture that the ground state of every frustration-free
Hamiltonian (with non-degenerate ground state) is a MPS. In fact,
the quantum Wielandt's inequality allows us to get a dichotomy
theorem in this direction:

\begin{theorem}
Take a local term $h$ with interaction length $L$ and assume that
$H_N=\sum_{i=1}^N \tau^i(h)$ is frustration-free and has a unique
ground state for every $N$. Its ground state can be represented as
an MPS with matrix size $D\times D$, where $D$ is
\begin{enumerate}
\item[(i)] either independent of $N$, \item[(ii)] or
$>\Omega(N^{\frac{1}{5}})$ for all prime numbers $N$.
\end{enumerate}
\end{theorem}
\begin{proof}
Let us recall from \cite[Theorem 5]{PVWC07} that each MPS with $D<N$
and $N$ prime can be mapped into a {\it canonical} decomposition
where all matrices are block diagonal $A_i=\oplus_{j=1}^b A_i^j$ and
each block satisfies injectivity. Moreover, \cite[Theorem
11]{PVWC07} states that if $b\ge 2$, $L_0=\max_j i(A^j)$ and $L$ is
the interaction length of any frustration-free translational
invariant Hamiltonian $H$ on $N$ spins having $|\psi_A\>$ as ground
state, the condition $N\ge 3(b-1)(L_0+1)+L$ implies that
$|\psi_{A^j}\>$ is also a ground state of $H$ for all $j$. Since the
quantum Wielandt's inequality allows us to bound $L_0\le O(D^4)$ and
trivially $b\le D$, we get that either (ii) $D\ge \Omega(N^\frac{1}{5})$,
or $b=1$ and $\ker(h)\ni \sum_{i_1,\cdots, i_L}{\rm
tr}(XA_{i_1}\cdots A_{i_L})|i_1\cdots i_L\>$ where $X\in
S_{N-L}(A)$. Since by the quantum Wielandt's inequality again
$N-L\ge i(A)$, we get that $\ker(h)\supseteq \{\sum_{i_1,\cdots,
i_L}{\rm tr}(XA_{i_1}\cdots A_{i_L})|i_1\cdots i_L\>: X\in
M_{D\times D}\}$. This trivially implies that $|\psi_A\>$ is also a
ground state for $H_{N'}$ when $N'>N$ and therefore the only one, so we
obtain (i).
\end{proof}

Regarding the restriction to prime $N$, note that by the Prime
Number Theorem the number of primes less than or equal to a given
$N$ is asymptotically $\frac{N}{\log N}$. Therefore in (ii) there are {\it many} lengths for which there is no MPS representation of the ground state with {\it small} matrices.

% SECTION
% SECTION

\section{Conclusions and open problems}

The present work focuses on finding dimension dependent bounds
for the number of times that a quantum channel has to be applied in
order to have a full-rank Choi matrix. Once this is obtained, bounds
on the quantum index of primitivity $q$ are straightforwardly
achieved, since $q \leq i(A)$. As direct applications of these
results, we derive dichotomy theorems for the zero-error capacity of
quantum channels as well as a couple of results in Matrix Product
States theory. The first one is the demonstration of a conjecture
with interesting implications for ground states of frustration-free
Hamiltonians and the other a theorem which introduces new
implications concerning the interaction-range of Parent
Hamiltonians.

As a possible future research, we suggest that it might be
advantageous to focus on computing bounds for $q$ directly, since $q
\neq i(A)$ in general. This is interesting because, while some
applications (like the ones in the MPS context) require bounds on
$i(A)$, others like Thm. \ref{thm:zero} are based on $q$. For
instance, from a purely mathematical point of view, $i(A)$ is not
applicable for positive maps (the usual framework of Frobenius
theory), unless the map is completely positive. Furthermore, we
leave  open the question about optimal bounds for both $q$ and
$i(A)$.

Another possible future research is to relate the quantum Wielandt's
inequality to graph theory and quantum random walks. In the
classical case, there is a close relationship between stochastic
matrices and graph theory (by taking $A$ the adjacency matrix of a
graph) which makes the inequality broadly applicable. In fact, the
usual proofs of the classical inequality are based on the graph
picture. However, although there are different attempts to establish
a relationship between quantum channels and quantum graphs
\cite{Kempe,BT07}, there is not any well-defined analogous one for
the quantum context.

% END OF TEXT
% END OF TEXT

\section*{Acknowledgment}

The authors would like to thank Frank Verstraete for his useful suggestions and discussions and A. Nogueira for the invaluable technical assistance.

\ifCLASSOPTIONcaptionsoff
  \newpage
\fi


\begin{thebibliography}{10}
\providecommand{\url}[1]{#1}
\csname url@samestyle\endcsname
\providecommand{\newblock}{\relax}
\providecommand{\bibinfo}[2]{#2}
\providecommand{\BIBentrySTDinterwordspacing}{\spaceskip=0pt\relax}
\providecommand{\BIBentryALTinterwordstretchfactor}{4}
\providecommand{\BIBentryALTinterwordspacing}{\spaceskip=\fontdimen2\font plus
\BIBentryALTinterwordstretchfactor\fontdimen3\font minus
  \fontdimen4\font\relax}
\providecommand{\BIBforeignlanguage}[2]{{%
\expandafter\ifx\csname l@#1\endcsname\relax
\typeout{** WARNING: IEEEtran.bst: No hyphenation pattern has been}%
\typeout{** loaded for the language `#1'. Using the pattern for}%
\typeout{** the default language instead.}%
\else
\language=\csname l@#1\endcsname
\fi
#2}}
\providecommand{\BIBdecl}{\relax}
\BIBdecl

\bibitem{Horn}
R.~A. Horn and C.~R. Johnson, \emph{Topics in Matrix Analysis}.\hskip 1em plus
  0.5em minus 0.4em\relax Cambridge University Press, 1991.

\bibitem{Wielandt}
H.~Wielandt, ``{Unzerlegbare, nicht negative Matrizen},'' \emph{Mathematische
  Zeitschrift}, vol.~52, no.~1, pp. 642--648, Dec 1950.

\bibitem{Seneta}
E.~Seneta, \emph{Non-negative matrices and Markov chains}, ser. Springer Series
  in Statistics.\hskip 1em plus 0.5em minus 0.4em\relax Springer-Verlag Berlin
  Heidelberg, 2006.

\bibitem{Frobenius}
J.~L.~R. Alfons\'in, \emph{The Diophantine Frobenius Problem}, ser. Oxford
  Lecture Series in Mathematics and its applications.\hskip 1em plus 0.5em
  minus 0.4em\relax Oxford University Press, 2005, vol.~30.

\bibitem{Varga}
R.~S. Varga, \emph{Matrix Iterative Analysis}, ser. Springer Series in
  Computational Mathematics.\hskip 1em plus 0.5em minus 0.4em\relax
  Springer-Verlag Berlin Heidelberg, 2000, vol.~27.

\bibitem{PVWC07}
D.~P\'erez-Garc\'ia, F.~Verstraete, M.~Wolf, and J.~Cirac, ``Matrix product
  state representations,'' \emph{Quantum Inf. Comput.}, vol.~7, p. 401, 2007.

\bibitem{gen-fro}
D.~E. Evans and R.~Hoegh-Krohn, ``Spectral properties of positive maps on
  c*-algebras,'' \emph{Journal of the London Mathematical Society}, vol. s2-17,
  no.~2, pp. 345--355, Apr 1978.

\bibitem{FNW92}
M.~Fannes, B.~Nachtergaele, and R.~F. Werner, ``{Finitely correlated states on
  quantum spin chains.}'' \emph{Commun. Math. Phys.}, vol. 144, no.~3, pp.
  443--490, 1992.

\bibitem{nilpotents}
M.~A. Fasoli, ``{Classification of nilpotent linear spaces in M(4; C)},''
  \emph{Communications in Algebra}, vol.~25, no.~6, pp. 1919 -- 1932, 1997.

\bibitem{Shannon}
C.~Shannon, ``The zero error capacity of a noisy channel,'' \emph{IRE Trans.
  Inform. Theory}, vol.~2, no.~3, pp. 8--19, Sep 1956.

\bibitem{zero-review}
J.~Korner and A.~Orlitsky, ``Zero-error information theory,'' \emph{IEEE Trans.
  Inform. Theory}, vol.~44, no.~6, pp. 2207--2229, Oct 1998.

\bibitem{Medeiros}
R.~A.~C. Medeiros and F.~M. de~Assis, ``Quantum zero-error capacity,''
  \emph{Int. J. Quant. Inf.}, vol.~3, no.~1, pp. 135--139, 2005.

\bibitem{Shor}
\BIBentryALTinterwordspacing
S.~Beigi and P.~W. Shor, ``{On the complexity of computing zero-error and
  Holevo capacity of quantum channels},'' Sep 2007. [Online]. Available:
  \url{http://arxiv.org/abs/0709.2090v3}
\BIBentrySTDinterwordspacing

\bibitem{Cubitt}
\BIBentryALTinterwordspacing
T.~S. Cubitt, J.~Chen, and A.~W. Harrow, ``Superactivation of the asymptotic
  zero-error classical capacity of a quantum channel,'' Jun 2009. [Online].
  Available: \url{http://arxiv.org/abs/0906.2547v2}
\BIBentrySTDinterwordspacing

\bibitem{Duan}
R.~Duan and Y.~Shi, ``Entanglement between two uses of a noisy multipartite
  quantum channel enables perfect transmission of classical information,''
  \emph{Phys. Rev. Lett.}, vol. 101, no.~2, p. 020501, Jul 2008.

\bibitem{Duan1}
\BIBentryALTinterwordspacing
R.~Duan, ``Super-activation of zero-error capacity of noisy quantum channels,''
  Jun 2009. [Online]. Available: \url{http://arxiv.org/abs/0906.2527v1}
\BIBentrySTDinterwordspacing

\bibitem{Lindblad}
G.~Lindblad, ``A general no-cloning theorem,'' \emph{Letters in Mathematical
  Physics}, vol.~47, no.~2, pp. 189--196, Jan 1999.

\bibitem{Frank-advances}
F.~Verstraete, V.~Murg, and J.~I. Cirac, ``Matrix product states, projected
  entangled pair states, and variational renormalization group methods for
  quantum spin systems,'' \emph{Advances in Physics}, vol.~57, no.~2, pp. 143
  -- 224, 2008.

\bibitem{Hastings}
M.~B. Hastings, ``An area law for one-dimensional quantum systems,''
  \emph{Journal of Statistical Mechanics: Theory and Experiment}, vol. 2007,
  no.~08, p. P08024, 2007.

\bibitem{PWSVC08}
D.~P\'erez-Garc\'ia, M.~M. Wolf, M.~Sanz, F.~Verstraete, and J.~I. Cirac,
  ``String order and symmetries in quantum spin lattices,'' \emph{Phys. Rev.
  Lett.}, vol. 100, no.~16, p. 167202, Apr 2008.

\bibitem{SWPC09}
M.~Sanz, M.~M. Wolf, D.~P\'erez-Garc\'ia, and J.~I. Cirac, ``Matrix product
  states: Symmetries and two-body hamiltonians,'' \emph{Phys. Rev. A}, vol.~79,
  no.~4, p. 042308, Apr 2009.

\bibitem{renormalization}
F.~Verstraete, J.~I. Cirac, J.~I. Latorre, E.~Rico, and M.~M. Wolf,
  ``Renormalization-group transformations on quantum states,'' \emph{Phys. Rev.
  Lett.}, vol.~94, no.~14, p. 140601, Apr 2005.

\bibitem{sequential}
C.~Sch\"on, E.~Solano, F.~Verstraete, J.~I. Cirac, and M.~M. Wolf, ``Sequential
  generation of entangled multiqubit states,'' \emph{Phys. Rev. Lett.},
  vol.~95, no.~11, p. 110503, Sep 2005.

\bibitem{sequential2}
L.~Lamata, J.~Le\'on, D.~P\'erez-Garc\'ia, D.~Salgado, and E.~Solano,
  ``Sequential implementation of global quantum operations,'' \emph{Phys. Rev.
  Lett.}, vol. 101, no.~18, p. 180506, Oct 2008.

\bibitem{SCP10}
\BIBentryALTinterwordspacing
N.~Schuch, I.~Cirac, and D.~P\'erez-Garc\'ia, ``Peps as ground states:
  degeneracy and topology,'' Jan 2010. [Online]. Available:
  \url{http://arxiv.org/abs/1001.3807v2}
\BIBentrySTDinterwordspacing

\bibitem{AKLT88}
I.~Affleck, T.~Kennedy, E.~H. Lieb, and H.~Tasaki, ``Valence bond ground states
  in isotropic quantum antiferromagnets,'' \emph{Communications in Mathematical
  Physics}, vol. 115, no.~3, pp. 477--528, Sep 1988.

\bibitem{Kempe}
J.~Kempe, ``Quantum random walks: an introductory overview,''
  \emph{Contemporary Physics}, vol.~44, no.~4, pp. 307 -- 327, 2003.

\bibitem{BT07}
\BIBentryALTinterwordspacing
A.~Ben-Aroya and A.~Ta-Shma, ``Quantum expanders and the quantum entropy
  difference problem,'' 2007. [Online]. Available:
  \url{http://arxiv.org/abs/quant-ph/0702129v3}
\BIBentrySTDinterwordspacing

\end{thebibliography}
\end{document}